\documentclass[11pt]{article}
\usepackage{amssymb}\usepackage{amsmath}
\usepackage{amsthm}
\newcommand{\be}{\begin{equation}}
\newcommand{\ee}{\end{equation}}

\newcommand{\tr}{\mathop{\rm tr}\nolimits}

\def\CS{{\mathcal S}}

\def\CT{{\mathcal T}}

\def\CE{{\mathcal E}}
\def\CO{{\mathcal O}}

\newcommand{\scp}[2]{\langle #1, #2 \rangle}
\theoremstyle{plain}
 
\newtheorem{thm}{\bfseries Theorem}
\newtheorem{propn}{\bfseries Proposition}
\newtheorem{lem}{\bfseries Lemma}
\newtheorem{cor}{\bfseries Corollary}
\theoremstyle{remark}
\newtheorem{rem}{\sc Remark} 
\newtheorem{ex}{\sc Example}
\begin{document}
\begin{center}
{\bf\large\sc 
Tensor space representations of Temperley--Lieb algebra \\[1mm]
 and generalized permutation matrices
} 

\vspace*{4mm}
{ Andrei Bytsko ${}^{1,2}$  }
\vspace*{3mm}

{\small
\noindent
${}^{1}$ Department of Mathematics, University of Geneva, 
C.P. 64, 1211 Gen\`eve 4, Switzerland \\ 
${}^{2}$ Steklov Institute of Mathematics,
Russian Academy of Sciences,
Fontanka 27, 191023, St.~Petersburg, Russia}
 
\end{center}
\vspace*{2mm}

\begin{abstract}
  
Orthogonal projections in ${\mathbb C}^n \otimes {\mathbb C}^n$ 
of rank one and rank two that give rise to unitary tensor space 
representations of the Temperley--Lieb algebra $TL_N(Q)$ are considered. 
In the rank one case, a complete classification of 
solutions is given. In the rank two case, solutions with $Q$
varying in the ranges $[2n/3,\infty)$ and $[n/\sqrt{2},\infty)$ are
constructed for $n=3k$ and $n=4k$, $k \in {\mathbb N}$, respectively.

\end{abstract}

\section{Introduction}
\subsection{Formulation of the problem and outline of results}

Below, we denote by $M_n$ the ring of $n\,{\times}\,n$ complex matrices,
by $I_n$ the $n\,{\times}\,n$ identity matrix, and by $\otimes$ the 
Kronecker product. $\bar{X}$, $X^t$, and $X^*$ stand, respectively, for 
the complex conjugate, the transpose, and the conjugate transpose 
of $X \in M_n$.

In the present article, we will continue the study begun 
in \cite{By1} of a particular class of representations of the 
Temperley--Lieb algebra $TL_N(Q)$. Recall that a unitary 
representation of $TL_N(Q)$ on the tensor product space
$\bigl({{\mathbb C}^n}\bigr)^{\otimes N}$ is determined by a matrix 
$T \in M_{n^2}$ satisfying the following relations:
\begin{align*} 
{}& (\mathrm{T1})     &&
  T^* = T ,&&  \\
{}&  (\mathrm{T2})  && 
	 T \, T = Q \, T   , && \\  
{}& (\mathrm{T3})     &&   
 T_{12} \, T_{23} \, T_{12} \, = T_{12}  \,,&&  \\
\label{tl2}
{}& (\mathrm{T4})   &&
 T_{23} \, T_{12} \, T_{23} \, = T_{23}  \,, &&    
\end{align*}
where 
$T_{12} \equiv T \,{\otimes}\, I_n$ and 
$T_{23} \equiv I_n \,{\otimes}\, T$. 
Without a loss of generality, we will always assume that $Q>0$.
Apart from $n$ and $Q$, an important parameter of a representation
is the rank $r=\mathrm{rank}(T)$.

\vspace*{1mm}{\small
\begin{ex}\label{ex1}
For $r=1$, the most known solution to (T1)--(T4) is given by 
\begin{equation}\label{xxz1}
 T = 
 \begin{pmatrix}
  0 & 0 & 0 & 0 \\
  0 & q & \zeta & 0 \\
  0 & \zeta^{-1} & q^{-1} & 0 \\
  0 & 0 & 0 & 0
 \end{pmatrix} 
 ,  \qquad\qquad q>0\,, \quad |\zeta|=1 \,.  
\end{equation}
The corresponding value of $Q$ in (T2) is $Q=q+q^{-1}$.
\end{ex}
}

The goal of the present article is to consider solutions
to  (T1)--(T4) in the cases $r=1$ and $r=2$. In the latter case,
our principal aim is to construct {\em varying $Q$ solutions} $T(q)$,
where $q$ is a parameter (or a set of parameters) and $Q = Q(q)$ 
is a non--constant function of~$q$ (like in Example~\ref{ex1}).  
It should be remarked here that, while rank one solutions
to (T1)--(T4) (and their non--Hermitian counterparts) are 
ubiquitous in the literature, the author is aware of only two
examples in the higher rank case --- see \cite{Kul3} and 
\cite{WX}, where two constructions are given for 
$r=n \geq 2$ but in both cases only for $Q=\sqrt{n}$.

The paper is organized as follows. In Section~\ref{RFOP},
we reformulate the original problem as a problem of constructing
a set of $r$ matrices $V_i$ satisfying an orthonormality condition
and such that the partitioned matrix $W_\CT$ built from them
is almost unitary.  In Section~\ref{RORO}, we give a complete
classification of rank one solutions by showing that every
suitable matrix $V \in M_n$ is unitarily congruent to a generalized
permutation matrix $DP_{\sigma}$, where $D$ is a non--singular
diagonal matrix and $\sigma$ is an
arbitrary involution which has at most one fixed point. 
In Section~\ref{ROT}, we focus on the rank two case,
where we have to find a suitable pair $V_1, V_2 \in M_n$.
In Section~\ref{PR1}, we establish some properties of  $V_1, V_2$. 
In Section~\ref{SN3}, we construct varying $Q$ solutions for 
$n=3p$,\ $p \in \mathbb N$ with $Q \in [2n/3,\infty)$.  
In particular, it is shown that every unitary matrix from $U(2p)$ 
gives rise to a solution to (T1)--(T4).
In Section~\ref{NEGP} and Section~\ref{SN4}, we consider 
the case when a solution is given by (or unitarily congruent to) 
a pair of generalized permutation matrices, i.e. 
$V_1 = D_1  P_{\sigma_1}$, $V_2 = D_2  P_{\sigma_2}$.
In Section~\ref{NEGP}, we establish some necessary conditions
for the pair $\sigma_1, \sigma_2$ and find all suitable pairs for $n=4$
and some for $n>4$. All these cases yield solutions
with $Q=n/\sqrt{2}$ and some of them admit varying~$Q$ solutions.
In Sections~\ref{SN4}, we construct varying $Q$
solutions for $n= 4l$,\ $l \in \mathbb N$ with $Q \in [n/\sqrt{2},\infty)$.
At the end of the section, we briefly discuss the extension
of constructed solutions to the non--Hermitian case 
corresponding to non--unitary representations of $TL_N(Q)$.
The proofs of all statements are given in the Appendix.
 
 \subsection{Reformulation of the problem}\label{RFOP}
 
 Let $\scp{\,}{}$ denote the standard inner product 
on~${\mathbb C}^n$ and let ${\CE}=\{e_a\}_{a=1}^n$ be 
a  basis of ${\mathbb C}^n$ orthonormal w.r.t. $\scp{\,}{}$.
Given a vector $v \in {\mathbb C}^n \otimes {\mathbb C}^n$,
we will write $v \sim V \in M_{n}$ if $V$ is  
the matrix of its coefficients, i.e.
$v  = \sum_{a,b=1}^n V_{ab}\, e_a \otimes e_b$.
Similarly, given an $r$--dimensional subspace $\CT \subset 
 {{\mathbb C}^n} \,{\otimes}\, {{\mathbb C}^n}$, we will write 
 $\CT \sim \{V_1,\ldots,V_r\}$
if the  orthonormal set of vectors, $v_1 \sim V_1, \ldots, v_r \sim V_r$,
is a spanning set of~$\CT$.
The corresponding orthogonal projection $P_\CT$ is represented 
by the following matrix: 
\begin{equation}\label{PiTau}
  P_\CT = 
 \sum_{s=1}^r \sum_{a,b,c,d=1}^n (V_s)_{ab} \, 
 (\bar{V}_s)_{cd} \  E_{ac} \otimes E_{bd} \,,
\end{equation}
where $E_{ab} \in {M}_{n}$ is such that
$\bigl(E_{ab}\bigr)_{ij} = \delta_{ai} \delta_{bj}$.

Every solution to (T1)--(T4) has the form $T=Q P_{\CT}$,
where $P_{\CT}$ is given by~(\ref{PiTau}). If $P_\CT$ has 
rank $r$, we will say, somewhat abusing the terminology,
that the corresponding representation is of rank~$r$. 

\vspace*{1mm}{\small
\begin{ex}
For $T$ given by (\ref{xxz1}), we have $T=(q+q^{-1})P_\CT$,
where $\CT \sim \{V\}$ and
\begin{equation}\label{vTq}
 V = \frac{1}{\sqrt{q^2+1}} 
 \biggl( \begin{matrix}
  0 & \zeta\, q \\
  1  & 0 
 \end{matrix} \biggr) \,, \qquad q>0 \,, \quad |\zeta|=1\,.
\end{equation} 
\end{ex}
}
\vspace*{1mm}

Given a subspace $\CT \sim \{V_1,\ldots,V_r\}$ of
${{\mathbb C}^n} \,{\otimes}\, {{\mathbb C}^n}$, 
 we associate to it the following partitioned matrix~$W_\CT \in M_{rn}$:
\begin{align}\label{WVdef} 
  W_{\CT}    = \sum_{s,m=1}^r  E_{sm} \otimes V_m \bar{V}_s \,.
\end{align}
Then we have the following criterion.

\begin{thm}[\cite{By1}, Theorem~2]\label{PROPTLW}
 $T=Q P_\CT \in M_{n^2}$, where $Q>0$ and $\CT \sim \{V_1,\ldots,V_r\}$,
 is a solution to $\mathrm{(T1)}$--$\mathrm{(T4)}$ if and only if 
$Q W_\CT$ is a unitary matrix.
\end{thm}

Thus, constructing a solution $T$ to (T1)--(T4) of 
a rank $r$ is equivalent to finding $r$ matrices, 
$V_1,\ldots,V_r$, such that the corresponding
vectors are orthonormal, i.e.
\begin{equation}\label{vv}
  \tr \bigl( V_s^*  V_m \bigr)  = \delta_{sm}  
\end{equation}
and the corresponding matrix $W_\CT$ is a scalar
multiple of  a unitary matrix.

It is natural to consider solutions $T$ and $T'$ as equivalent
if the corresponding sets $V_1,\ldots,V_r$ and $V'_1,\ldots,V'_r$
are related by  simultaneous {\em unitary congruence}: 
\begin{equation}\label{VVkg} 
 V'_k = g \, V_k \, g^t , \qquad  k=1, \ldots,r\,,  \qquad g \in U(n) \,,
\end{equation}  
because, as seen from (\ref{PiTau}), such $T$ and $T'$
are related as follows
\begin{equation}\label{Tkg0}
 T' =  (g \otimes g) \, T \, (g^* \otimes g^*) \,.
\end{equation}

In this context, it is useful to recall the 
following criterion of unitary congruence:
\begin{lem}[\cite{HoHo}, Theorem~2.4]\label{ABUC}  
Non--singular matrices $A, B \in M_n$ are unitarily congruent 
if and only if there exists a unitary matrix $g \in U(n)$ such that
\begin{equation}\label{abg}
A\,A^* = g\, B\,B^* g^* , \qquad
A\,\bar{A} = g\, B\,\bar{B}\, g^* .
\end{equation}
\end{lem}

\section{Representations of rank one}\label{RORO}

Here we consider solutions to (T1)--(T4) such that
$r \equiv \mathrm{rank} (T)=1$.

Let $V \in M_n$ satisfy the normalization condition
\begin{equation}\label{trvv} 
\tr \bigl( V^* V \bigr) =1 \,.
\end{equation} 
By Theorem~\ref{PROPTLW},  $T=Q P_\CT$,
where $Q>0$ and $\CT \sim \{V\}$, is a solution to (T1)--(T4)
if and only if  
\begin{equation}\label{VV}     
   Q \,  V \bar{V} \in U(n) \,,
\end{equation}  
or, equivalently, 
\begin{equation}\label{VV22}     
    V \bar{V} V^t V^* = Q^{-2} I_n \,.
\end{equation}   

\begin{rem}\label{REun}
For every $V$ satisfying (\ref{trvv}) and (\ref{VV}), we have
$Q \geq n$ (see Proposition~3 in \cite{By1}). The lower bound,
 $Q=n$, is achieved only if $V$ itself is {\em almost unitary},
 i.e. it is a scalar multiple of a unitary matrix. 
\end{rem}

\begin{rem}
In the rank one case, solving equations (T2)--(T4) without imposing 
the hermiticity condition (T1) amounts to solving the following 
counterpart of equation (\ref{VV22}):
$VUV^tU^t =Q^{-2} I_n$, where $V$ and $U$ are non--singular matrices 
such that $\tr(VU^t) =1$.  A scheme of construction of suitable 
pairs $V, U$ was outlined in~\cite{Gur}. Particular solutions, $U=Q^{-1}V^{-1}$ and $U=Q^{-1}(V^t)^{-1}$, were considered 
in \cite{Kul} and~\cite{WX}, respectively. 
Note that, in the latter case, the only possible value of $Q$ 
is $Q=n$. This solution is a counterpart of the almost unitary 
solution to (\ref{VV}) mentioned in Remark~\ref{REun}.
\end{rem}
 
Let us introduce some notations. $\CS_n$ will stand for
the symmetric group of degree~$n$. If we need to write
down the explicit form of a permutation $\sigma \in \CS_n$,
we will give its decomposition into cycles.
Given an element $\sigma$ of $\CS_n$, we will denote by
$P_{\sigma} \in M_n$  the corresponding permutation matrix, i.e.
$(P_\sigma)_{ij} = \delta_{i,\sigma(j)}$.
Matrix $P_\sigma^t$ corresponds to $\sigma^{-1}$.
If $\sigma$ is an involution, i.e. $\sigma^{-1}=\sigma$,
then $P_\sigma$ is a symmetric matrix.
Given a diagonal matrix $D \in M_n$, we will denote
by $D^\sigma \equiv P_\sigma D P_\sigma^t$ the matrix 
obtained from $D$ by the action of the permutation $\sigma$ 
on its diagonal entries, i.e.  for the matrix entries we have:
$\bigl(D^\sigma)_{k,k}=D_{\sigma^{-1}(k),\sigma^{-1}(k)}$.

\vspace*{0.5mm}
The most general form of a solution to equations  
(\ref{trvv})--(\ref{VV}) is the following.

\begin{thm}\label{VV0} 
Let $\sigma \in \CS_n$ be an involution which has at most
one fixed point and $P_\sigma$ be the corresponding 
permutation matrix.
For every $V \in M_n$ satisfying (\ref{trvv}) and (\ref{VV}), 
there exists $g \in U(n)$ such that
$V' = g \, V \, g^t$ has the following form:
\begin{equation}\label{VDP} 
   V' =  D \, P_\sigma \,, 
\end{equation} 
with $D={\rm diag}(z_1,\ldots,z_n)$, where
$z_k \in {\mathbb C}{\setminus}\{0\}$ satisfy the following relations:
\begin{equation}\label{VDP3} 
   \sum_{k=1}^n |z_k|^2 = 1 \,, \qquad\qquad
  z_k \, z_{\sigma(k)}  = Q^{-1}  \quad \forall k \,.
\end{equation}
\end{thm}

In other words, any solution to (\ref{trvv})--(\ref{VV})
is unitarily congruent to a {\em generalized permutation matrix} 
$D \, P_\sigma$, where
 $P_\sigma$ is an a priori chosen permutation matrix such that
\begin{equation}\label{Pperm} 
   P_\sigma^t=P_\sigma \,, 
   \qquad\qquad \tr P_\sigma = n \!\!\pmod 2 \,,
\end{equation} 
and the diagonal matrix $D$ satisfies the 
following relations:
\begin{equation}\label{VDP2} 
   \tr D \bar{D} = 1 \,, \qquad
 D^{-1} = Q \, D^{\sigma} .
\end{equation} 

The proof of Theorem~\ref{VV0} is given in the Appendix.
Here we remark only that the proof simplifies if the spectrum of 
$W_\CT \equiv V \bar{V}$ is assumed to be simple. 
In the general case, the proof is based on the results of
\cite{HLV} and \cite{HoSe} on normal
forms of {\em congruence normal} matrices.

Theorem~\ref{VV0} along with equations (\ref{VVkg}) and (\ref{Tkg0})
allows us to describe all solutions to (T1)--(T4) in the 
rank one case as follows.

\begin{cor}\label{Tgg}
Let $\{e_a\}_{a=1}^n$ be the canonical
orthonormal basis of ${\mathbb C}^n$.
For every permutation $\sigma \in \CS_n$ which is an involution 
and has at most one fixed point and for every $T \in {M}_{n^2}$ which
has rank one and satisfies relations (T1)--(T4) with $Q>0$,  there exists
 a unitary matrix $g \in U(n)$ such that 
\begin{equation}\label{Tg0}
 T = Q \, (g \otimes g) \, (v \otimes v^*) \, 
 	(g^* \otimes g^*) \,,
\end{equation}
where $ v =  
  \sum_{k=1}^{n} 
    z_k \, e_k \otimes e_{\sigma(k)} $
with $z_k \in {\mathbb C}\setminus\{0\}$ satisfying relations~(\ref{VDP3}).
\end{cor}

\begin{rem}\label{ZN}
Consider equations (\ref{VDP3}) for $\sigma=(1,n)(2,n-1)\ldots$.
Then, given $z_1,{\ldots},z_{\lfloor\frac n2 \rfloor}$, we can obtain  
the value of $Q$ and then find the remaining $z$'s
(up to the sign of $z_{\frac {n+1}2 }$ in the case when
$n$ is odd). Thus, a generic solution to (\ref{trvv})--(\ref{VV}) 
is determined by $\lfloor \frac n2 \rfloor$ complex parameters.  
A solution to  (\ref{trvv})--(\ref{VV}) which is unitarily
congruent to $V'$ of the form (\ref{VDP}), where
$P_\sigma$ does not satisfy one or both conditions~(\ref{Pperm}),
will be degenerate, that is, it will have fewer parameters.
\end{rem} 

\vspace*{1mm}{\small
\begin{ex}
For $n=2$, the group $\CS_2$ consists of two elements,
$\sigma=id$ and $\sigma=(12)$. $P_{(12)}$ fulfils  
conditions (\ref{Pperm}).  So, by Theorem~\ref{VV0}, the general 
solution to (\ref{trvv})--(\ref{VV}) is unitarily congruent to
\begin{equation}\label{VTq}  
  V'(u) = \frac{1}{\sqrt{Q}} \, 
  {\rm diag}(u,u^{-1})\, P_{(12)}
 = \frac{1}{\sqrt{Q}}
  \begin{pmatrix}
  0 &   u \\
   u^{-1} & 0
 \end{pmatrix} \,, \qquad
 Q= |u|^2 + |u|^{-2} \,.
\end{equation}
In accordance with Remark~\ref{ZN}, the general solution has 
one complex parameter, $u \in {\mathbb C}\setminus \{0\}$. 

Looking for a solution built using $P_{id}$ instead of $P_{(12)}$, 
we obtain a degenerate solution: 
\begin{equation}\label{VTq0}  
 V_0= 
 \frac{1}{\sqrt{Q_0}} \, {\rm diag}(u_0,u_0) P_{id}= 
 \frac{u_0}{\sqrt{Q_0}} \,
 \begin{pmatrix}
  1 & 0 \\ 0 & 1 
 \end{pmatrix} \,, \qquad u_0 = \pm 1 \,, \qquad Q_0=2.
\end{equation}
Since Theorem~\ref{VV0} states that every $n=2$ solution 
is unitarily congruent to (\ref{VTq}), $V_0$ must
be unitarily congruent to $V'(u_0)=u_0 P_{(12)}/\sqrt{Q_0}$. 
Indeed, Lemma~\ref{ABUC}
assures that $P_{(12)}$ and $P_{id}$ are unitarily congruent.
To establish this unitary congruence explicitly, one can verify 
the following equality:
\begin{equation}\label{PPg12}  
 g_0 \,
 \begin{pmatrix}
  0 &   1 \\
   1 & 0
 \end{pmatrix} 
 g_0^t = 
 \begin{pmatrix}
  1 &   0 \\
   0 & 1
 \end{pmatrix} \,, \qquad\quad
 g_0= \frac{e^{-i\pi/4}}{\sqrt{2}} 
 \begin{pmatrix}
  1 &   i \\
   i & 1
 \end{pmatrix} \,.
\end{equation}
Thus, $g_0 V'(u_0) \, g_0^t = V_0$.
\end{ex}
} 
\vspace*{1mm}
 
\begin{rem} 
$V_1, V_2 \in M_n$ that satisfy (\ref{trvv})--(\ref{VV}) 
for the same value of~$Q$ are not necessary unitarily congruent.
Indeed,  by Theorem~\ref{VV0}, they are unitarily congruent,
respectively, to $V_1'=D_1P_{\sigma}$ and $V_2'=D_2P_{\sigma}$, 
where $\sigma$ is an involution and $D_1, D_2$ satisfy~(\ref{VDP2}).
Lemma~\ref{ABUC} implies that the sets of singular values
of unitarily congruent matrices coincide. Thus, a necessary 
condition for $V_1'$ and $V_2'$ (and hence for $V_1$ and $V_2$) to 
be unitarily congruent to each other is that $D_1 \bar{D}_1$ and 
$D_2 \bar{D}_2$ coincide up to a permutation of their entries.
For $n \geq 4$, among solutions to (\ref{VDP2}) there are
pairs $D_1, D_2$ that do not satisfy this condition.
\end{rem}

\section{Representations of rank two}\label{ROT}

\subsection{Preliminary remarks}\label{PR1}

In the rest of article, we will consider solutions 
to (T1)--(T4) such that
$r \equiv \mathrm{rank} (T)=2$.
 
Let $V_1, V_2 \in M_n$ be such that 
\begin{equation}\label{norm2} 
 \tr \bigl( V_1^* V_1 \bigr) = 
 \tr \bigl( V_2^* V_2 \bigr) = 1 \,, \qquad
 \tr \bigl( V_1^* V_2 \bigr) =0 \,. 
\end{equation}   
Set
\begin{align}\label{WV2} 
  W_{\CT}  \equiv   \left( 
 \begin{matrix}
   V_1 \bar{V}_1 & V_2 \bar{V}_1   \\
     V_1 \bar{V}_2 & V_2 \bar{V}_2  \end{matrix} \right).
\end{align}
By Theorem~\ref{PROPTLW}, $T=Q P_\CT$, $\CT \sim \{V_1,V_2\}$,
is a solution to (T1)--(T4) iff $Q W_\CT$ is a unitary matrix, 
which is equivalent to the following set of equations:
\begin{align}
\label{veq1} 
{}&   V_1 \bar{V}_1 V_1^t V_1^* +
    V_2 \bar{V}_1 V_1^t V_2^* = Q^{-2} I_n \,,\\
\label{veq2} 
{}&   V_1 \bar{V}_2 V_2^t V_1^* +
    V_2 \bar{V}_2 V_2^t V_2^* = Q^{-2} I_n \,,\\
\label{veq3} 
{}&   V_1 \bar{V}_1 V_2^t V_1^* +
    V_2 \bar{V}_1 V_2^t V_2^* = 0 \,.
\end{align}
It is worth noting that, unlike the rank one case,
matrices $V_1$ and $V_2$ can be singular. 

\begin{propn}\label{VVdet} 
If $V_1, V_2 \in M_n$ satisfy (\ref{veq1})--(\ref{veq3}) for some $Q$,
then $|\det V_1| =|\det V_2|$. Furthermore, 
both $V_1$ and $V_2$ are singular if $n$ is odd. 
\end{propn} 

\begin{rem} 
For every pair $V_1, V_2 \in M_n$ satisfying (\ref{norm2}) and 
(\ref{veq1})--(\ref{veq3}), we have 
 \begin{align}\label{Q2n} 
  Q= \sqrt{2} \ \ \text{if} \ \ n=2\,, \qquad
  Q \geq \frac{n}{2} \ \ \text{if} \ \ n \geq 3\,,
\end{align}
see Theorem~3 and Corollary~3 in~\cite{By1}. 
Furthermore, $Q=n/\sqrt{2}$ if either $V_1$ or $V_2$ is
a scalar multiple of a unitary matrix, 
cf. Proposition~6 in~\cite{By1}.
\end{rem}

The condition that $Q W_\CT$ be unitary implies that each
block $QV_i\bar{V}_j$ is a contraction. If at least one of
the blocks is itself a scalar multiple of a unitary matrix, then
the estimate (\ref{Q2n}) sharpens as follows.

\begin{propn}\label{VVG2}
Let $V_1, V_2 \in M_n$ satisfy (\ref{veq1})--(\ref{veq3}) for some $Q>0$.
Suppose, in addition, that  $\alpha V_1 \bar{V}_1$ is unitary for some
 $\alpha>0$ and $ \tr \bigl( V_1^* V_1 \bigr) = 1$.
Then\\[1mm]
i) $\alpha V_2 \bar{V}_1$, 
 $\alpha V_1 \bar{V}_2$, and  $\alpha V_2 \bar{V}_2$ are unitary.\\[1mm]
ii) There exist  $g,g' \in U(n)$ such that $V_2= V_1 \, g$
and $V_2= g' \, V_1$. \\[1mm]
iii) $V_1, V_2$ satisfy (\ref{norm2}). \\[1mm]
iv) $\alpha = \sqrt{2} Q$ and
\begin{equation}\label{Qguni3}   
  Q \geq \frac{n}{\sqrt{2} } \,.
\end{equation} 
\end{propn} 

Examples of rank two solutions, where all the blocks of $Q W_\CT$
are almost unitary, will be given in Sections 3.3 and 3.4

\subsection{Solutions for $n= 3p$}\label{SN3}

Here we will construct rank two solutions in the
case when $n$ is a multiple of~$3$.
Consider the following ansatz:
\begin{equation}\label{n3part} 
  V_1= \begin{pmatrix}
  0 & F_{11} & 0 \\
  \bar{G}_{11} & 0 & \bar{G}_{12} \\
  0 & F_{21} & 0 
\end{pmatrix} 
 , \qquad 
 V_2= \begin{pmatrix}
  0 & F_{12} & 0 \\
 \bar{G}_{21} & 0 & \bar{G}_{22} \\
  0 & F_{22} & 0 
\end{pmatrix} ,
\end{equation}
where $F_{ij}, G_{ij} \in M_p$, $p \geq 1$.
Note that ${\det V_1 = \det V_2 =0}$, so that
the ansatz is consistent with Proposition~\ref{VVdet} for
all~$p$.

\begin{thm}\label{FGn3}
For $p \in {\mathbb N}$, let $\alpha_1, \alpha_2$ be some positive 
numbers such that 
\begin{equation}\label{betatr} 
\frac{1}{\alpha_1^2}  + \frac{1}{\alpha_2^2} = \frac{1}{p} \,.
\end{equation}
Suppose that $F_{ij}, G_{ij} \in M_p$ are such that the following 
partitioned matrices 
\begin{equation}\label{Hpart} 
  H_1=  \alpha_1 \, \begin{pmatrix}
  F_{11} & F_{12} \\ 
  F_{21} & F_{22}
\end{pmatrix} 
 , \qquad 
 H_2=  \alpha_2 \, \begin{pmatrix}
  G_{11} & G_{12} \\ 
  G_{21} & G_{22} 
\end{pmatrix} 
\end{equation}
are unitary. 
If $p>1$, suppose, in addition, that the equality
\begin{equation}\label{gamHH} 
  \zeta \, \bigl( G_{11} \, F_{12} + G_{12} \, F_{22}  \bigr) =
   G_{21} \, F_{11} + G_{22} \, F_{21}
\end{equation}
holds for some $\zeta \in {\mathbb C}$ such that $|\zeta|=1$.

Then $V_1, V_2 \in M_{3p}$ given by (\ref{n3part}) satisfy 
relations (\ref{norm2}) and (\ref{veq1})--(\ref{veq3}) with
\begin{equation}\label{QGH} 
 Q = \alpha_1 \, \alpha_2  \,,
\end{equation}
and, therefore, $T=Q P_\CT$, $\CT \sim\{V_1,V_2\}$ 
is a solution to (T1)--(T4).
\end{thm} 
 
\begin{rem}
Condition (\ref{betatr}) implies the following inequality
for $Q$ given by (\ref{QGH}):
\begin{equation}\label{QGH2} 
 Q \geq 2p = \frac{2}{3} n  \,.
\end{equation}
\end{rem}

\begin{propn}\label{FGex1} 
Let $p \in {\mathbb N}$ and let positive
$\alpha_1, \alpha_2$ satisfy (\ref{betatr}).\\
i) The hypotheses of Theorem~\ref{FGn3} are fulfilled if
\begin{equation}\label{FG0} 
 \alpha_1 \, F_{ij} = U_{ij} \,, \qquad  
 \alpha_2 \, G_{ij} = U^*_{ji} \,,  
\end{equation}
providing that the partitioned matrix 
$U = \begin{pmatrix}
  U_{11} & U_{12} \\ 
  U_{21} & U_{22} 
\end{pmatrix}$
is unitary.\\
ii) The hypotheses of Theorem~\ref{FGn3} are fulfilled if 
\begin{equation}\label{FG1} 
 F_{12} =  -w\, F_{22} \,, \quad 
  F_{21} = \bar{w}\,  F_{11}   \,, \quad
  G_{12} =  w\, G_{11} \,, \quad 
  G_{21} =  -\bar{w}\, G_{22} \,,  
  \qquad w \in {\mathbb C} \,,
\end{equation}
providing that $\beta_1\,F_{11}$, $\beta_1\, F_{22}$,
$\beta_2\,G_{11}$, $\beta_2\, G_{22}$
are unitary for $\beta_i= \alpha_i \sqrt{1+|w|^2}$. 
\end{propn} 

{\small
\begin{ex} 
The pair 
\begin{equation}\label{soln3} 
  V_1= \begin{pmatrix}
  0 & z_1 & 0 \\
  z_2 & 0 & \bar{w} z_2 \\
  0 & \bar{w} z_1 & 0 
\end{pmatrix} 
\,, \qquad 
 V_2= \begin{pmatrix}
  0 & -\zeta_1 w z_1 & 0 \\
  -\zeta_2 w z_2 & 0 &  \zeta_2 z_2 \\
  0 &  \zeta_1  z_1 & 0 
\end{pmatrix}  ,
\end{equation}
where $z_1, z_2, w, \zeta_1, \zeta_2 \in {\mathbb C}$ and
\begin{equation}\label{3z00}
 |\zeta_1| =|\zeta_2|=1 \,, \qquad 
 |z_1|\,|z_2| \neq 0 \,, \qquad 
  (|z_1|^2 + |z_2|^2)(1 + |w|^2) =1 \,,
\end{equation}
satisfies (\ref{norm2}) and (\ref{veq1})--(\ref{veq3}) with 
\begin{equation}\label{Q3z00}  
  Q^{-1} =  |z_1|\,|z_2| \, (1 + |w|^2)  \,.
\end{equation}
In particular, setting  $w=0$, $\zeta_2=-\zeta_1=1$, and
$z_1=(q^4+1)^{- \frac{1}{2}}$,\   $z_2=q^2 z_1$, $q >0$, 
we recover Example~13 {}from \cite{By1} constructed as 
a TL pair for the quantum algebra $U_q(su_2)$.
\end{ex} 
}

\begin{rem} 
For $p=1$, condition (\ref{gamHH}) follows from the
hypothesis that $H_1, H_2$ are unitary (indeed, if  $A=H_2 H_1 \in U(2)$,
then  $|A_{12}|=|A_{21}|$ which is equivalent to~(\ref{gamHH})).
Therefore, taking two generic elements from $U(2)$ as $H_1, H_2$, 
we obtain the following solution.
\end{rem} 

{\small
\begin{ex}\label{V4z}
The pair
\begin{equation}\label{soln4} 
  V_1= \begin{pmatrix}
  0 & z_1 & 0 \\
  z_2 & 0 & z_3 \\
  0 & z_4 & 0 
\end{pmatrix} 
\,, \qquad 
 V_2= \begin{pmatrix}
  0 & - \zeta_1 \bar{z}_4 & 0 \\
  -\zeta_2 \bar{z}_3 & 0 & \zeta_2 \bar{z}_2 \\
  0 & \zeta_1 \bar{z}_1 & 0 
\end{pmatrix}  ,
\end{equation}
where $z_1, z_2, z_3, z_4, \zeta_1, \zeta_2 \in {\mathbb C}$ and 
\begin{equation}\label{4z0} 
 |\zeta_1| =|\zeta_2|=1 , \quad 
 |z_1| + |z_4| \neq 0 , \quad 
 |z_2| + |z_3| \neq 0 , \quad 
  |z_1|^2 + |z_2|^2 + |z_3|^2 + |z_4|^2 =1 ,
\end{equation}
satisfies (\ref{norm2}) and (\ref{veq1})--(\ref{veq3}) with
\begin{equation}\label{Q4z0} 
  Q^{-2}= (|z_1|^2 + |z_4|^2)(|z_2|^2 + |z_3|^2) \,.
\end{equation}
\end{ex}
}

Employing generalized permutation matrices, we will
construct a solution generalizing Example~\ref{V4z}
for which (\ref{gamHH}) holds non--trivially (that is,
unlike for the cases given in Proposition~\ref{FGex1}, 
the l.h.s. and the r.h.s. of (\ref{gamHH}) 
do not vanish identically). 

\begin{propn}\label{FGex2}
Given $p \in {\mathbb N}$ and $\sigma_1, \sigma_2 \in {\mathcal S}_p$,
let $P_{\sigma_1}, P_{\sigma_2}$ be the corresponding
permutation matrices and let $F_{ij}, G_{ij} \in M_p$ be given by
\begin{eqnarray}
& F_{11}= P_{\sigma_1} D_1 , \quad 
F_{12}= -P_{\sigma_1} \bar{D}_4 Z_1, \quad
 F_{21}= P_{\sigma_1} D_4 , \quad 
 F_{22}= P_{\sigma_1} \bar{D}_1 Z_1 , & \\
 & G_{11}= D_2 P_{\sigma_2} , \quad G_{12}= D_3 P_{\sigma_2}  , \quad
 G_{21}= - Z_2 \bar{D}_3 P_{\sigma_2} , 
 \quad G_{22}= Z_2  \bar{D}_2 P_{\sigma_2}  , &
\end{eqnarray}
where $D_i \in M_p$ are diagonal matrices and $Z_i \in U(p)$
are diagonal unitary matrices. Then the hypotheses of 
Theorem~\ref{FGn3} are satisfied providing that
\begin{equation}\label{Daa} 
  \alpha_1^2 (D_1 \bar{D}_1 + D_4 \bar{D}_4) = I_p =
  \alpha_2^2 (D_2 \bar{D}_2 + D_3 \bar{D}_3)  
\end{equation}
for some positive $\alpha_1, \alpha_2$ satisfying (\ref{betatr}) and
\begin{equation}\label{ZZss} 
  \zeta \, M \, Z_1^{\sigma_2 \circ \sigma_1} = 
  Z_2 \, \bar{M} \,,
\end{equation}
where 
$M \equiv  ( D_3 \bar{D}_1^{\sigma_2 \circ \sigma_1} -
  D_2 \bar{D}_4^{\sigma_2 \circ \sigma_1})$ and $|\zeta|=1$.
\end{propn}

\subsection{Generalized permutations matrices, 
solutions for $Q=n/\sqrt{2}$}\label{NEGP}

For $n$ even, we will look for solutions to 
(\ref{veq1})--(\ref{veq3}) that can be brought 
by a simultaneous unitary congruence (\ref{VVkg})
to a pair of generalized permutation matrices,
\begin{equation}\label{vvddpp}   
  V_1 = D_1  P_{\sigma_1} \,, \qquad 
   V_2 = D_2  P_{\sigma_2} \,,
\end{equation}
where $\sigma_1, \sigma_2 \in {\mathcal S}_n$ and
$D_1, D_2 \in M_n$ are non--singular diagonal matrices.
Since pairs related by (\ref{VVkg}) are regarded as
equivalent and all permutation matrices are unitary, 
we will search for pairs of the form (\ref{vvddpp})
up to the transformations
\begin{equation}\label{sstt}   
 \sigma_1 \to \tau \circ \sigma_1  \circ \tau^{-1} \,, \qquad
 \sigma_2 \to \tau \circ \sigma_2  \circ \tau^{-1} \,, \qquad  
 \tau \in \CS_n\,.
\end{equation}

For $V_1, V_2$ of the form (\ref{vvddpp}), 
conditions (\ref{norm2}) turn into
\begin{eqnarray}
\label{trdd12a}   
 &   \tr D_1 \bar{D}_1 = \tr D_2 \bar{D}_2 = 1 \,, & \\
\label{trdd12b} 
 &  \tr \bigl( D_1 \bar{D}_2 \, 
   P_{\sigma_1} P_{\sigma_2}^t \bigr) =0 \,, &
\end{eqnarray}
and equations (\ref{veq1})--(\ref{veq3}) are equivalent
to the following system:
\begin{align}
\label{Deq1} 
{}&   D_1 \bar{D}_1 D_1^{\sigma_1}  \bar{D}_1^{\sigma_1}  +
    D_2 \bar{D}_2 D_1^{\sigma_2} \bar{D}_1^{\sigma_2} 
    = Q^{-2} I_n \,,\\
\label{Deq2} 
{}&  D_1 \bar{D}_1 D_2^{\sigma_1}  \bar{D}_2^{\sigma_1} + 
    D_2 \bar{D}_2 D_2^{\sigma_2} \bar{D}_2^{\sigma_2} 
    = Q^{-2} I_n \,,\\
\label{Deq3} 
{}&  D_1 \bar{D}_1^{\sigma_1} \, P_{\sigma'} \, D_2^{\sigma_1} \bar{D}_1 + 
  D_2 \bar{D}_1^{\sigma_2} \, P_{\sigma''} \, 
  D_2^{\sigma_2} \bar{D}_2  =0 \,,
\end{align}
where $\sigma'=\sigma_1 \circ \sigma_1 \circ 
 \sigma_2^{-1} \circ \sigma_1^{-1}$
and $\sigma''=\sigma_2 \circ 
  	\sigma_1 \circ \sigma_2^{-1} \circ \sigma_2^{-1}$.
Since $D_1, D_2$ are non--singular, equation (\ref{Deq3}) requires 
that $\sigma'=\sigma''$ that is 
\begin{equation}\label{ssss}   
  \sigma_2^{-1} \circ \sigma_1 \circ 
  \sigma_1 \circ \sigma_2^{-1}  
  =  \sigma_1 \circ \sigma_2^{-1} \circ \sigma_2^{-1} 
  \circ \sigma_1  .
\end{equation}   	

Below, we will write $\sigma_1 \asymp \sigma_2$ if
$\sigma_1$ and $\sigma_2$ are commuting permutations.
For instance, (\ref{ssss}) is equivalent to the condition 
$\sigma_1 \circ \sigma_2^{-1} \asymp 
\sigma_2^{-1} \circ  \sigma_1$.

For $n$ even, let us call $\sigma_1, \sigma_2 \in \CS_n$ 
an {\em admissible} pair of permutations if 
a) they satisfy relation (\ref{ssss}); 
b) they have no common fixed points;
c) $\sigma_2^{-1} \circ \sigma_1$ has no fixed points
if $\sigma_1$ and $\sigma_2$ 
are involutions or if $\sigma_1 \asymp \sigma_2$. 
Two admissible pairs are regarded as equivalent if they are
related by the transformation (\ref{sstt}) combined,
if necessary, with the exchange 
$\sigma_1 \leftrightarrow \sigma_2$.

\begin{lem}\label{FixP}
If a pair $\sigma_1, \sigma_2 \in \CS_n$ is not 
admissible, then equation (\ref{Deq3}) has no solution 
for non--singular diagonal matrices $D_1, D_2 \in M_n$.
\end{lem}

Note that Lemma~\ref{FixP} excludes, in particular, 
the case $\sigma_1=\sigma_2$.

\begin{rem}
If $\sigma_2^{-1} \circ \sigma_1$ has no fixed points, 
then condition (\ref{trdd12b}) is satisfied for any $D_1, D_2$.
\end{rem}

\begin{propn}\label{SSn4} 
i) Every admissible pair $\sigma_1, \sigma_2 \in \CS_4$  
is equivalent to one of the
pairs in the following list: 
\begin{align*}
& a)\ \sigma_1= id,\ \sigma_2= (12)(34); &
& b)\ \sigma_1= id,\ \sigma_2= (1234);\\
& c)\ \sigma_1=(1)(23)(4),\  \sigma_2= (14)(2)(3);  &
& d)\ \sigma_1=(1)(23)(4),\  \sigma_2= (1342);\\
& e)\ \sigma_1= (1)(23)(4),\  \sigma_2=(12)(34); &
& f)\ \sigma_1=(1)(234),\  \sigma_2= (321)(4);\\
& g)\ \sigma_1= (1234),\  \sigma_2= (13)(24); &
& h)\ \sigma_1= (1234),\  \sigma_2= (12)(34);\\
& i)\ \sigma_1= (1234),\  \sigma_2= (4321);  &
& j)\ \sigma_1= (12)(34),\  \sigma_2= (14)(23).
\end{align*}
ii) For every admissible pair $\sigma_1, \sigma_2$
in this list, there exist vectors
$\vec{u}, \vec{v} \in {\mathbb R}^4$  such that matrices 
$V_1 = \frac{1}{2} \mathrm{diag}
(e^{i\pi u_1},\ldots, e^{i\pi u_4})\, P_{\sigma_1}$ and
$V_2 = \frac{1}{2} \mathrm{diag}
(e^{i\pi v_1},\ldots, e^{i\pi v_4})\, P_{\sigma_2}$ 
satisfy (\ref{norm2}) and (\ref{veq1})--(\ref{veq3})
with $Q=2\sqrt{2}$.
\end{propn} 

\begin{rem}
$\sigma_2^{-1} \circ \sigma_1$ has no fixed points
for all the pairs in the list except the case~h). 
\end{rem}

Let us say that $\sigma_1, \sigma_2 \in \CS_n$  have complementary 
sets of fixed points if for every $k=1,\ldots,n$, we have either $\sigma_1(k)=k$
and $\sigma_2(k) \neq k$ or $\sigma_2(k)=k$ and $\sigma_1(k) \neq k$. 
The cases a), b), and c) are of this type and they admit 
the following generalization.

\begin{propn}\label{sscomp}
For $n$ even, let $\sigma_1, \sigma_2 \in \CS_n$  be 
composed only of 1--cycles (corresponding to fixed points) 
and cycles of even length. If such $\sigma_1, \sigma_2$ 
have complementary sets of fixed points, then
$\sigma_1, \sigma_2$ is an admissible pair and
there exist diagonal matrices $D_1, D_2 \in M_n$ such that
$V_1 =D_1 P_{\sigma_1}$ and $V_2 =D_2 P_{\sigma_2}$
satisfy (\ref{norm2}) and (\ref{veq1})--(\ref{veq3})
with $Q=n/\sqrt{2}$.
\end{propn}

We will see below that the cases h), i), and j)
and their generalizations to greater $n$ divisible by
four allow us to construct representations of rank two
not only for $Q=n/\sqrt{2}$ but for 
$Q$ varying in the range $[n/\sqrt{2},\infty)$.

\subsection{Generalized permutations matrices, 
varying $Q$ solutions for $n= 4l$}\label{SN4}

Observe that relation (\ref{ssss}) holds if
$\sigma_1, \sigma_2$ satisfy  the following conditions:
\begin{equation}\label{zzeq0}   
 \sigma_1 \asymp \sigma_2 \circ \sigma_2 \quad
 \text{and} \quad 
\sigma_2 \asymp \sigma_1 \circ \sigma_1 \,.
\end{equation} 
For such $\sigma_1, \sigma_2$, we have
\begin{equation}\label{zzz}   
  \sigma' =\sigma''  =\sigma_2^{-1} \circ \sigma_1 ,
\end{equation}
and equation (\ref{Deq3}) acquires the following form:
\begin{equation}\label{dddd}   
  D_1^{\sigma_2} \bar{D}_1^{\sigma_2 \circ \sigma_1} 
  D_2^{\sigma_1 \circ \sigma_1} \bar{D}_1^{\sigma_1} =
 - D_2^{\sigma_2} \bar{D}_1^{\sigma_2 \circ \sigma_2} 
  D_2^{\sigma_1 \circ \sigma_2} \bar{D}_2^{\sigma_1}  .
\end{equation}
Note that (\ref{zzeq0}) holds, in particular, if $\sigma_1$ 
and $\sigma_2$ commute or if they both are involutions.  
 
 {\small
\begin{ex}\label{troiss}
For $n$ even, the following pairs $\sigma_1, \sigma_2 \in \CS_n$ 
satisfy~(\ref{zzeq0}):
\begin{align}
\label{ssigma1}  
{}&  \sigma_1=  (1,\ldots,n) \,, \qquad
  \sigma_2= (n,\ldots,1) , \\
\label{ssigma2}     
{}&  \sigma_1=  (1n)(2,n\,{-}\,1)\ldots(\frac{n}{2},\frac{n}{2}\,{+}\,1) , \qquad
  \sigma_2=  (12)(34)\ldots(n-1,n) .
\end{align}
They are admissible, respectively, for  
$n= 2l+2$,\ $l \in \mathbb N$
and  $n= 4l$,\ $l \in \mathbb N$.
For $n=4$, (\ref{ssigma1}) and (\ref{ssigma2}) recover,
respectively, the cases i) and j) in Proposition~\ref{SSn4}.
\end{ex}
}

\begin{thm}\label{DDSS} 
For $n$ even, let $\sigma_1, \sigma_2 \in \CS_n$ satisfy (\ref{zzeq0})
and let $\sigma_2^{-1} \circ \sigma_1$ have no fixed points. 
Let  $P_{\sigma_1}, P_{\sigma_2}$ be the corresponding
permutation matrices and let $A, B \in M_n$ be given by
\begin{equation}\label{AB1}  
  A= \bigl( I_n + P_{\sigma_2} \bigr)
   \bigl( P_{\sigma_2} - P_{\sigma_1} \bigr) \,, \qquad
 B= \bigl( I_n + P_{\sigma_1} \bigr)
   \bigl( P_{\sigma_1} - P_{\sigma_2} \bigr) \,.
\end{equation}
Let $\vec{x} \in {\mathbb R}^n$ be a vector such that
\begin{equation}\label{ppx}  
 P_{\sigma_1} \, \vec{x} =
  P_{\sigma_2} \, \vec{x} = -\vec{x} \,.
\end{equation}
Suppose that there exit vectors 
$\vec{u}, \vec{v} \in {\mathbb R}^n$  
such that all the components of the vector
\begin{equation}\label{ABsub}  
  \vec{w} = A\,\vec{u} + B\, \vec{v} 
\end{equation}
are odd integers. 

Let $D_1, D_2 \in M_n$ be diagonal matrices such that
\begin{equation}\label{d12unitr}  
 (D_1)_{kk} = \mu^{-1} e^{x_k+i \pi u_k} , \qquad
 (D_2)_{kk} = \mu^{-1} e^{x_k+i \pi v_k},
\end{equation}
where $\mu^2= \sum_{k=1}^n e^{2x_k}$.

Then $V_1 =D_1 P_{\sigma_1}, V_2 =D_2 P_{\sigma_2}$ satisfy 
relations (\ref{norm2}) and (\ref{veq1})--(\ref{veq3}) 
with $Q$ given by
\begin{equation}\label{Qdd}  
  Q = \frac{1}{\sqrt{2}} \, \sum_{k=1}^n e^{2x_k} ,
\end{equation}
and, therefore, $T=Q P_\CT$, $\CT \sim\{V_1,V_2\}$ 
is a solution to (T1)--(T4). 
\end{thm} 

\begin{rem}\label{vvu}
Condition (\ref{ppx}) implies that
$D_i \bar{D}^{\sigma_i}_i$ is a multiple of a unitary matrix 
and hence so is $V_i \bar{V}_i$. Therefore, by
Proposition~\ref{VVG2}, we have $Q \geq n/\sqrt{2}$.
The value  $Q=n/\sqrt{2}$ is achieved only if $\vec{x}= \vec{0}$, 
in which case matrices $V_1, V_2$ are themselves 
almost unitary.
\end{rem}

For the admissible pairs given in Example~\ref{troiss}, 
vector $\vec{x}=(x,-x,x,-x,\ldots)$ satisfies (\ref{ppx}) and, 
moreover, condition (\ref{ABsub}) turns out to be resolvable 
if $n$ is divisible by four. Let us write 
$D= \mathrm{diag}_m (d_1,\ldots,d_m)$ if $D$ is a diagonal 
matrix such that $(D)_{k+m,k+m}=(D)_{k,k}$.

\begin{propn}\label{VVn4k} 
For $n=4l$, $l \in \mathbb N$, let
 $\sigma_1, \sigma_2 \in {\mathcal S}_n$ be given by
either  (\ref{ssigma1}) or (\ref{ssigma2}) and 
let $D_1, D_2 \in M_{n}$ be given by
\begin{equation}\label{D4n} 
 D_1 = \mathrm{diag}_4 (z_1,z_2,z_1,z_2) \,, \qquad
 D_2 = \mathrm{diag}_4 (z_1,-\zeta z_2,z_1, \zeta z_2) \,, 
\end{equation}
where $z_1, z_2, \zeta \in {\mathbb C}$ are such that
\begin{equation}\label{zzet1} 
 |\zeta|=1 \,, \qquad |z_1|\,|z_2| \neq 0 \,, \qquad
 |z_1|^2 + |z_2|^2  = \frac{2}{n} \,. 
\end{equation}
Then the pair $V_1=D_1 P_{\sigma_1}$,
$V_2=D_2 P_{\sigma_2}$ satisfies (\ref{norm2}) and (\ref{veq1})--(\ref{veq3}) with 
\begin{equation}\label{Q4n} 
   Q = \frac{1}{\sqrt{2}\,|z_1|\,|z_2|} \,.
\end{equation}
\end{propn} 

\vspace*{1mm}{\small
\begin{ex}
For $n=4$,\  $V_1, V_2$ corresponding to $\sigma_1, \sigma_2$
given by (\ref{ssigma1}) look as follows: 
\begin{equation}\label{soln4a} 
  V_1=   \begin{pmatrix}
  0 & 0 & 0 & z_1 \\
  z_2 & 0 & 0 & 0  \\
  0& z_1 & 0 & 0 \\
  0 & 0 & z_2 & 0
\end{pmatrix} 
\,, \qquad 
 V_2=   \begin{pmatrix}
    0 & z_1 & 0 & 0 \\
    0 & 0 & -\zeta\, z_2 & 0 \\
    0 & 0 & 0 & z_1 \\
     \zeta\, z_2 & 0 & 0 & 0
\end{pmatrix}  
\end{equation} 
and $V_1$, $V_2$ corresponding to $\sigma_1, \sigma_2$
given by (\ref{ssigma2}) are 
\begin{equation}\label{soln4b} 
  V_1=   \begin{pmatrix}
  0 & 0 & 0 &  z_1\\
  0 & 0 &  z_2 & 0    \\
  0 & z_1 & 0 & 0 \\
   z_2 & 0 & 0 & 0 
\end{pmatrix}  
\,, \qquad 
 V_2=   \begin{pmatrix}
  0 &  z_1 & 0 & 0\\
  -\zeta\, z_2 & 0 & 0 & 0    \\
  0 & 0 & 0 & z_1 \\
  0& 0 & \zeta\, z_2 & 0
\end{pmatrix} 
\end{equation} 
In both cases, $|z_1|^2 + |z_2|^2  =1/2$ and $|z_1||z_2|\neq 0$.
\end{ex}
}

A pair $\sigma_1, \sigma_2 \in {\mathcal S}_n$ can be admissible despite that $\sigma_2^{-1} \circ \sigma_1$ has fixed points.
The case h) in Proposition~\ref{SSn4} is an example of such 
a pair. It can be generalized as follows.

\vspace*{1mm} {\small
\begin{ex}\label{troiss2}
For $n= 2l+2$,\ $l \in \mathbb N$, the following pair 
$\sigma_1, \sigma_2 \in \CS_n$ is admissible and satisfies~(\ref{zzeq0}):
\begin{align} 
\label{ssigma3}     
{}&  \sigma_1=  (1,\ldots,n) \,, \qquad
  \sigma_2= (12)(34)\ldots(n-1,n) .
\end{align}
\end{ex}
}
\vspace*{1mm}

\begin{propn}\label{VVn4}
For $n=4l$, $l \in \mathbb N$, let
 $\sigma_1, \sigma_2 \in {\mathcal S}_n$ be given by
(\ref{ssigma3}) and let $D_1, D_2 \in M_{n}$ be given by 
\begin{equation}\label{D4ns} 
D_1 = \mathrm{diag}_4 (z_1,z_2,z_1,z_3) \,, \qquad
 D_2 = \mathrm{diag}_4 (z_1,-\zeta \bar{z}_3,z_1, \zeta \bar{z}_2) \,, 
\end{equation}
where $z_1, z_2, z_3, \zeta \in {\mathbb C}$ are such that
\begin{equation}\label{zzet2} 
|\zeta|=1 , \qquad |z_1| \neq 0 , \qquad
|z_2| + |z_3| \neq 0, \qquad
 2|z_1|^2 + |z_2|^2 + |z_3|^2 = \frac{4}{n} \,.  
\end{equation} 
Then the pair $V_1=D_1 P_{\sigma_1}$,
$V_2=D_2 P_{\sigma_2}$ satisfies (\ref{norm2}) and (\ref{veq1})--(\ref{veq3}) with 
\begin{equation}\label{Q4ns} 
   Q = \frac{1}{   |z_1| \sqrt{|z_2|^2+|z_3|^2 } } \,.  
\end{equation}
\end{propn}

\vspace*{1mm}{\small
\begin{ex}
For $n=4$,\  $V_1$, $V_2$ look as follows:
\begin{equation}\label{soln4c} 
  V_1=   \begin{pmatrix}
  0 & 0 & 0 & z_1 \\
  z_2 & 0 & 0 & 0  \\
  0& z_1 & 0 & 0 \\
  0 & 0 & z_3 & 0
\end{pmatrix} 
 , \qquad 
 V_2=   \begin{pmatrix}
  0 & z_1 & 0 & 0\\
  -\zeta\, \bar{z}_3 & 0 & 0 & 0    \\
  0 & 0 & 0 & z_1 \\
  0& 0 & \zeta\, \bar{z}_2 & 0 
\end{pmatrix}  ,
\end{equation}
where $2|z_1|^2 + |z_2|^2 + |z_3|^2 =1$ and 
$|z_1|(|z_2| + |z_3|) \neq 0$.
\end{ex}
}
\vspace*{1mm}

\begin{rem}
In Proposition~\ref{VVn4}, $V_1 \bar{V}_1$ and $V_2 \bar{V}_2$ 
are not multiples of unitary matrices unless $|z_2|=|z_3|$.
Furthermore, unlike the case  of Proposition~\ref{VVn4k},
we can set either $z_2=0$ or $z_3=0$. In which case,
both $V_1$ and $V_2$ become degenerate in accordance with
Proposition~\ref{VVdet}.
\end{rem}

\begin{rem}
Despite that $V_i \bar{V}_i$ in  Proposition~\ref{VVn4}
are not in general almost unitary, $Q$ given by (\ref{Q4ns})
satisfies the same inequality, $Q \geq n/\sqrt{2}$,
as in the case of Proposition~\ref{VVn4k}.
\end{rem}

 \begin{rem}
The three pairs of matrices $V_1$, $V_2 \in M_{4l}$ constructed 
in Proposition~\ref{VVn4k} and Proposition~\ref{VVn4} are 
not unitarily congruent to each 
other for generic values of $z_1, z_2, z_3$. Indeed, 
$\chi(V)\equiv \tr \bigl(V \bar{V} \bigr)$ is invariant
under unitary congruence. But for generic $z_1, z_2, z_3$, 
we have $\chi(V_1)=\chi(V_2)=0$ if $\sigma_1, \sigma_2$
are given by (\ref{ssigma1}),
$\chi(V_1) \neq 0$, $\chi(V_2) \neq 0$
if $\sigma_1, \sigma_2$ are given by (\ref{ssigma2}), 
and $\chi(V_1) = 0$, $\chi(V_2) \neq 0$
if $\sigma_1, \sigma_2$ are given by~(\ref{ssigma3}). 
\end{rem}

\begin{rem}
It was pointed out in Remark~1 of \cite{By1} that certain varying $Q$ 
solutions to (T1)--(T4) extend to non--Hermitian solutions 
to (T2)--(T4) thus extending corresponding unitary tensor space 
representations of $TL_N(Q)$ to non--unitary ones.
For instance,  $T$ given in Example~\ref{ex1} remains a solution 
to (T2)--(T4) for $q,\zeta \in {\mathbb C}\setminus \{0\}$. 
Below, we give an example for the rank two case.
\end{rem}

\vspace*{1mm} {\small
\begin{ex}
For $n=4l$, $l \in {\mathbb N}$ and the representations constructed 
in Proposition~\ref{VVn4k},  parametrize $z_1$ and $z_2$ as follows 
(cf. equation~(\ref{vTq})):  
\begin{equation}\label{zpar}
z_1= \sqrt{\frac{2}{n}}\, 
 \frac{q\,\xi_1}{\sqrt{q^2+1}} \,, \qquad
z_2= \sqrt{\frac{2}{n}}\, 
  \frac{\xi_2}{\sqrt{q^2+1}} \,, \qquad
 q>0 \,, \quad |\xi_1|=|\xi_2|=1 \,.
\end{equation}
By (\ref{Q4n}), we have $Q=n\sqrt{2}(q+q^{-1})/4$ and 
it is evident from (\ref{PiTau}) that entries of 
$T(q,\xi_1,\xi_1,\zeta) =Q P_{\CT}$  
are rational functions in $q, \xi_1, \xi_2, \zeta$
with a pole at the origin of the complex plain. Therefore,
equalities (T2)--(T4) imply that certain rational functions 
in these variables vanish identically and hence these equalities 
remain valid for  $q, \xi_1, \xi_2, \zeta \in {\mathbb C}\setminus \{0\}$.
\end{ex}
 }
\vspace*{1mm} 

The representations constructed in Proposition~\ref{VVn4} 
extend to non--unitary ones in the same vein.

 \section*{Appendix}

The proof of Theorem~\ref{VV0} will be preceded
by the following Lemma.

\begin{lem}\label{eigenW}
If $V \in M_n$ satisfies (\ref{VV}) with $Q>0$ and
$W_\CT \equiv V \bar{V}$, then the following holds:\\[1mm] 
i)  $V$ is non--singular. $\det (Q W_\CT) =1$.\\[1mm]
ii) The set of singular values of $V$ comprises
$\lfloor \frac{n}{2} \rfloor$ pairs of the form 
$(\lambda_k,\lambda'_k)$, where 
$\lambda_k\, \lambda'_k = Q^{-1}$, and if $n$ is odd, 
one unpaired singular value equal to $Q^{- \frac 12}$. \\[1mm]
iii) The set of eigenvalues of $QW_\CT$ comprises
$\lfloor \frac{n}{2} \rfloor$ pairs of the form 
$(\zeta_k,\bar{\zeta}_k)$, where $|\zeta_k|=1$,  
and if $n$ is odd, one unpaired eigenvalue equal to unity. 
\end{lem} 
 
\begin{proof}[\bf Proof of Lemma~\ref{eigenW}] 
$i)$ We have $\det (QW_\CT) = Q^n |\det V|^2$ and also
$|\det (QW_\CT) | =1$ since $Q W_\CT$ unitary.
Hence $\det V \neq 0$ and $\det (QW_\CT) =1$ because $Q>0$. \\
$ii)$ Let $\lambda_1 \geq \lambda_2 \geq {\ldots} \geq \lambda_n$
be the set of singular values of $V$.
Then the set of eigenvalues of $V^* V$ and 
$\overline{V V^*}$ is $\{\lambda_1^2,{\ldots},\lambda_n^2\}$. 
Equation (\ref{VV22}) can be rewritten in the
form $Q^2 \overline{V V^*} = (V^* V)^{-1}$.
Which implies that $Q^2\{\lambda_1^2,{\ldots},\lambda_n^2\}$
coincides with $\{\lambda_n^{-2},{\ldots},\lambda_1^{-2}\}$,
i.e. $\lambda_k \lambda_{n+1-k} = Q^{-1}$.
If $n$ is odd, we have $\lambda_{(n+1)/2}^2=Q^{-1}$. \\
$iii)$ Note that $\bar{W}_\CT = \bar{V} V = V^{-1} W_\CT V$.
Hence, if $\zeta \neq \pm 1$ is an eigenvalue of $QW_\CT$ and
the corresponding eigenspace is spanned by vectors $x_1,\ldots,x_m$,
then the eigenspace corresponding to $\bar{\zeta}$ 
is spanned by $V\bar{x}_1,\ldots,V\bar{x}_m$.
It remains to note that the eigenspace corresponding to
$\zeta=-1$ is even--dimensional because, by $i)$,
$\det (QW_\CT) =1$.
\end{proof}

\begin{proof}[\bf Proof of Theorem~\ref{VV0}]  
Let $V \in M_n$ satisfy (\ref{trvv}) 
and (\ref{VV}) and let $W_\CT \equiv V \bar{V}$.  

Let us first prove Theorem~\ref{VV0} assuming that
the spectrum of $W_\CT$ is simple. In this case,
taking Lemma~\ref{eigenW} into account, it follows that there exists 
$g\in U(n)$ such that $\tilde{W}= Q\, g W_\CT g^*$ is 
a diagonal unitary matrix such that 
\begin{equation}\label{gen} 
  \tilde{W}_{aa} \, \tilde{W}_{bb} = 1 
  \qquad \text{if and only if} \qquad 
  b = \sigma(a) \,,
\end{equation} 
where $\sigma \in {\mathcal S}_n$ is an involution
that has no fixed points if $n$ is even and one fixed point
if $n$ is odd. Note that $\tilde{W}^\sigma=\tilde{W}^{-1}$.
For $V_0 = g \, V \, g^t$, we have
\begin{equation}\label{VVD} 
 Q \, V_0 \bar{V}_0 =  \tilde{W} , \qquad
 Q \, \bar{V}_0 V_0 =  \tilde{W}^{-1} .  
\end{equation}  
Whence we conclude that $V_0 = \tilde{W} \, V_0 \, \tilde{W}$
and thus  $(V_0\bigr)_{ab} = 
\tilde{W}_{aa} \tilde{W}_{bb} \bigl(V_0\bigr)_{ab}$.
Taking (\ref{gen}) into account, we infer that 
$(V_0)_{ab}=0$ unless $b=\sigma(a)$.
That is, we have established that $V$ is unitarily 
congruent to $V_0 = D_0 P_{\sigma}$, where $D_0$ is a diagonal
 matrix which, by (\ref{VVD}), satisfies the equation 
$Q D_0 \bar{D}_0^{\sigma} = \tilde{W}$.
The general solution to this equation is 
$D_0= T\,w\,H$, where $T={\rm diag}(t_1,t_2,\ldots)$,
$t_k >0$ and $H={\rm diag}(\xi_1,\xi_2,\ldots)$,
$|\xi_k|=1$ are such that $Q T^{\sigma} = T^{-1}$ and
$H^{\sigma} = H$, whereas $w$ is a diagonal unitary matrix 
such that $w^2=\tilde{W}$ and $w^{\sigma} = w^{-1}$.
Let $h$ be a unitary diagonal matrix such 
that $h^2=H$ and $h^{\sigma} = h$.
Then $V_0'= h^{-1} V_0 h^{-1}$ is unitarily 
congruent to $V_0$ (and hence to $V$)
and we have $V_0'= D P_{\sigma}$, where
$D =T\,w$. Observe that
$ Q \, D  D^{\sigma} = 
Q \, T \, w \, T^{\sigma} w^{\sigma} = I_n$. 
Thus, $D$ satisfies the second equation in~(\ref{VDP2}).
The first equation in (\ref{VDP2}) for $D$ follows from 
the condition~(\ref{trvv}).

\vspace*{1mm}
In order to prove Theorem~\ref{VV0} in the general case,
i.e. without assuming anything about the spectrum of $W_\CT$, 
we will invoke some results about {\em congruence normal}\
matrices obtained first in \cite{HLV} and developed further 
in \cite{HoSe} (see also \cite{IkFa} and Problem 4.4.P41
in~\cite{HoJo}). 

Recall that $A \in M_n$ is called congruence normal if
$A \bar{A}$ is normal.

\begin{lem}[\cite{HoSe}, Theorem~5.3]\label{AbA} 
If $A$ is a non--singular congruence normal matrix,
then $A$ is unitarily congruent to a block--diagonal
matrix, where each block is of the form
\begin{equation}\label{block}   
  \bigl(s\bigr)  
  \quad \text{or} \quad
  \begin{pmatrix}
  0 &   \mu\,t \\
  t & 0
 \end{pmatrix} ,\qquad
  s,t >0 \,,\ \mu \in {\mathbb C}\setminus\{0,1\} .
\end{equation}
\end{lem}

\begin{rem}
For $\mu=1$, the 2$\times$2 block in (\ref{block}) is 
unitarily congruent to $t I_2$, cf.~eq.~(\ref{PPg12}).
\end{rem}

Now, let $V \in M_n$ satisfy (\ref{VV}) and 
$A \equiv Q^{\frac 12 } V$.
Since the r.h.s. of (\ref{VV}) is a unitary matrix,
$A$ is non--singular and congruence normal.
Therefore, by Lemma~\ref{AbA}, $A$ is unitarily 
congruent to a block--diagonal
matrix with blocks as in~(\ref{block}).
Hence $A \bar{A}$ is unitarily similar  
to a block--diagonal matrix with blocks
$\bigl(s_k^2\bigr)$ and $t_k^2 \, {\rm diag}(\mu_k,\bar{\mu_k})$.
Note that $s_k = |\mu_k| \, t_k^2 =1$ because
$A \bar{A}$ is unitary. So, $A$ is unitarily congruent to
$A'={\rm diag}(1,\ldots,1,B_1,B_2,\ldots)$, where
$B_k = \bigl(\begin{smallmatrix}
  0 & \zeta_k/t_k \\
  t_k & 0
 \end{smallmatrix}\bigr)$, $t_k >0$, $|\zeta_k| =1$.
By Lemma~\ref{ABUC}, ${\rm diag}(1,1)$ and
$B_k$ are unitarily congruent, respectively, to
$\bigl(\begin{smallmatrix}
  0 & 1\\
  1 & 0
 \end{smallmatrix}\bigr)$
and 
$\tilde{B}(z_k) =\bigl(\begin{smallmatrix}
  0 & z_k\\
  z_k^{-1} & 0
 \end{smallmatrix}\bigr)$,
where $z_k=\zeta_k^{1/2}/t_k \in {\mathbb C}\setminus\{0\}$.
Thus, $A'$ (and hence $A$) is unitarily congruent to  
$A'' = {\rm diag}(\tilde{B}(z_1),\ldots,
\tilde{B}(z_{\lfloor \frac{n}{2} \rfloor}),1)$,
where some $z_k$'s can be equal to unity and the 
last unity block is present if $n$ is odd.
So, $V$ is unitarily congruent to
$V''=Q^{-1/2}A''=D P_{\sigma_0}$, where 
$\sigma_0 = (12)(34)\ldots$ and 
$D=Q^{-1/2}\text{diag}(z_1,1/z_1,z_2,1/z_2,\ldots)$.
Clearly, we have
\begin{equation}\label{dds0}   
   Q \, D \, D^{\sigma_0} = I_n .
\end{equation}
That is, $D$ satisfies the second equation in~(\ref{VDP2}).
The first equation in (\ref{VDP2}) for $D$ follows from 
the condition~(\ref{trvv}).

It remains to note that every involution 
$\sigma \in {\mathcal S}_n$ that has the same number of fixed 
points as $\sigma_0$ does can be constructed as
$\sigma = \tau \circ \sigma_0 \circ \tau^{-1}$
by choosing a suitable $\tau \in {\mathcal S}_n$.
Therefore $V$ is unitarily congruent to 
$V'''= P_\tau D P_{\sigma_0} P^t_\tau = D' P_\sigma$, 
where $D'=D^\tau$. Obviously, we have 
$\tr D \bar{D} = \tr D' \bar{D}'$.
And applying the permutation $\tau$ to (\ref{dds0}), we
verify the second equation in~(\ref{VDP2}) for $D'$: 
$I_n= Q D^\tau (D)^{\tau \circ\sigma_0}
=Q D' (D')^\sigma$.
\end{proof}

\begin{proof}[\bf Proof of Proposition~\ref{VVdet}]   

Recall Jacobi's identity for submatrices  (see, e.g. 
eq.~(0.8.4.2) in \cite{HoJo}): 
if $A$ and $A^{-1}$ are partitioned matrices,
$A=\Bigl(\begin{smallmatrix} A_{11} & A_{12} \\ 
A_{21} & A_{22}\end{smallmatrix}\Bigr)$, 
$A^{-1}=\Bigl(\begin{smallmatrix} A'_{11} & A'_{12} \\ 
A'_{21} & A'_{22}\end{smallmatrix}\Bigr)$
and the blocks $A_{11}$ and $A'_{11}$ are of the
same size, then $\det A'_{22} = (\det A_{11}) / \det A $.
Take $A=QW_\CT$ and $A_{11} = Q V_1 \bar{V}_1$.
Since $A$ is unitary, we have $A^{-1}=A^*$
and so $A'_{22}=Q V_2^t V_2^*$.
Invoking Jacobi's identity and taking into account
that $|\det A|=1$, we infer that
$|\det V_1|^2= |\det (A_{11}/Q)| = |\det (A'_{22}/Q)| = |\det V_2|^2$.
Thus, $|\det V_1|= |\det V_2|$ and hence $V_1$ and 
$V_2$ are either both singular or both non--singular.  
 
Rewriting equation (\ref{veq3}) in the form 
$V_1 \bar{V}_1 V_2^t V_1^* = -
    V_2 \bar{V}_1 V_2^t V_2^*$
and comparing the determinants of the both sides,
we infer that $\det(\bar{V}_1 V_2^t) |\det V_1|^2
 = (-1)^n  \det(\bar{V}_1 V_2^t) |\det V_2|^2$.
Whence it follows that, if $n$ is odd, $V_1$ and $V_2$
cannot be both non--singular.  Taking into account that 
$|\det V_1|= |\det V_2|$, we conclude that
both $V_1$ and $V_2$ are singular.
\end{proof}

\begin{proof}[\bf Proof of Proposition~\ref{VVG2}]   
$ii)$ Eq. (\ref{veq1}) was derived from the condition 
$Q^2 W_{\CT} W^*_{\CT} =I_{2n}$. 
Its counterpart derived from the equivalent 
condition $Q^2 W^*_{\CT} W_{\CT} =I_{2n}$ reads
\begin{equation}\label{veq1b}    
   V_1^t V_1^* V_1 \bar{V}_1  +
    V_2^t V_1^* V_1 \bar{V}_2  = Q^{-2} I_n \,.
\end{equation} 
Since $\alpha V_1 \bar{V}_1$ is unitary, $V_1$ is  non--singular
and we have $\alpha^2 \bar{V}_1 V_1^t =  ( V_1^* V_1 )^{-1}$.
Using this relation, we rewrite (\ref{veq1}) and (\ref{veq1b}) as follows:
\begin{eqnarray}
\label{veq5}   
&&  (V_2 V_1^{-1}) (V_2 V_1^{-1})^* = (\alpha^2 Q^{-2} -1) I_n \,, \\
 \label{veq5b}    
&&  (V_1^{-1} V_2)^t ( \overline{V_1^{-1} V_2}) = 
  (\alpha^2 Q^{-2} -1) I_n \,. 
\end{eqnarray} 
Taking Proposition~\ref{VVdet} into account, we infer that
the l.h.s. of (\ref{veq5}) and (\ref{veq5b}) are positive definite matrices
and their determinants are equal to unity.  Therefore, 
$\alpha^2 Q^{-2} -1 =1$ and thus $\alpha =\sqrt{2}Q$.
Furthermore, (\ref{veq5}) and (\ref{veq5b}) imply, respectively, that
$V_2 = g' \, V_1$ and $V_2 = V_1 g$ where $g$ and $g'$ are unitary.

$i)$ As a consequence of $ii)$, we have
$\alpha V_1 \bar{V}_2 = \alpha V_1 \bar{V}_1 \bar{g}$,
$\alpha V_2 \bar{V}_1 = \alpha g' V_1 \bar{V}_1$,
$\alpha V_2 \bar{V}_2 = \alpha g' V_1 \bar{V}_1 \bar{g}$,
and so all these matrices are unitary.

$iii)$ Another consequence of $ii)$ is
$\bar{V}_2 V_2 ^t = \bar{V}_1 V_1^t$
and hence $\tr V_2 V_2^* = \tr V_1 V_1^* =1$.
Multiplying (\ref{veq3}) with $V_1^{-1}$ from the left and with $(V^*_1)^{-1}$
{}from the right, taking trace, and taking into
account that  $V_1^{-1} V_2 =g$ is unitary, we conclude
that $\tr \bar{V}_1 V_2^t=0$.

$iv)$ The inequality on $Q$ is implied by Proposition~7 from \cite{By1}
for $r=2$.
\end{proof} 

\begin{proof}[\bf Proof of Theorem~\ref{FGn3}]  
For $V_1, V_2$ given by (\ref{n3part}), the corresponding  
matrix $W_\CT$  can be brought to a block diagonal form 
by permutations of its block--rows and block--columns. 
Indeed, let $P_1$ and $P_2$ be the permutation matrices
corresponding to the permutations $\{123456\} \to \{143625\}$
and $\{123456\} \to \{134625\}$, respectively. Then we have
\begin{equation}\label{Wpp}   
 \bigl( P_1 \otimes I_p) \, W_{\CT} \, \bigl( P_2^t \otimes I_p) =
 \begin{pmatrix}
  W_1 & 0 \\ 
  0 & \bar{W}_2
\end{pmatrix}  \,,
\end{equation} 
where $W_1 \in M_{4p}$ and $W_2 \in M_{2p}$ are given by
\begin{equation}\label{W1}   
  \alpha_1 \alpha_2 \,  W_1 = 
    H_1 \otimes_{\scriptscriptstyle 2 \times 2} H_2
\end{equation} 
and
\begin{equation}\label{W2}   
  W_2 =  \begin{pmatrix}
  G_{11} F_{11} + G_{12} F_{21} & 
  G_{21} F_{11} + G_{22} F_{21} \\ 
  G_{11} F_{12} + G_{12} F_{22} & 
  G_{21} F_{12} + G_{22} F_{22}
\end{pmatrix} .
\end{equation} 
In (\ref{W1}), the Kronecker product is understood as that
for $2\,{\times}\,2$ matrices $H_1, H_2$ that have 
noncommuting entries
$\alpha_1 F_{ij}$, $\alpha_2 G_{ij}$. In other words, we have
\begin{equation}\label{W1b}   
  W_1 = \sum_{a,b,c,d =1} ^2
  E_{ab} \otimes E_{cd} \otimes F_{ab} \, G_{cd} \,,
\end{equation}
where $E_{ab}$ are the basis $2\,{\times}\,2$ matrices,
i.e. $(E_{ab})_{ij}=\delta_{ai}\delta_{bj}$.

Equation (\ref{Wpp}) implies that $Q W_\CT$ is unitary
iff $QW_1$ and $QW_2$ are unitary. By (\ref{W1b}),
we have
\begin{equation}\label{W1c}   
 W_1 \, W_1^* =  \sum_{a,b,c,d,i,j =1} ^2
  E_{ab} \otimes E_{cd} \otimes F_{ai} \, G_{cj} 
  \, G^*_{dj} \, F^*_{bi} \,.
\end{equation}
Therefore, if $H_1, H_2$ are unitary, that is if  
the following relations hold:
\begin{equation}\label{FfGg}   
  \alpha_1^2 \sum_{b=1}^2 F_{ab}  \, F^*_{cb} = 
   \alpha_1^2 \sum_{b=1}^2  F^*_{ba}  \, F_{bc}=
  \delta_{ac}  I_p =
  \alpha_2^2 \sum_{b=1}^2 G_{ab}  \, G^*_{cb} =
  \alpha_2^2 \sum_{b=1}^2 G^*_{ba}  \, G_{bc} ,
\end{equation}
we infer from (\ref{W1c}) that $ \alpha_1  \alpha_2 W_1$
is unitary.

Next, if relation (\ref{gamHH}) holds, we can rewrite
$W_2$ given by (\ref{W2}) in the following form:
\begin{equation}\label{W2b}   
  \alpha_1  \alpha_2 \, W_2 =  S \, H_2  H_1 \, \bar{S}\,, \qquad
  S={\rm diag} (\zeta^{\frac 12}, \bar{\zeta}^{\frac 12}) 
  \otimes I_p \,.
\end{equation}
Whence it is evident that  $ \alpha_1  \alpha_2 W_2$
is unitary if  $H_1, H_2$ are unitary.

Finally, we note that relations (\ref{norm2})  for $V_1, V_2$ 
given by (\ref{n3part}) acquire the following form:
\begin{equation}\label{trFG}   
   \tr V_i V_j^* = \sum_{a=1}^2  \tr F_{ai} F^*_{aj} +
   \sum_{a=1}^2  \tr \bar{G}_{ia} G^t_{ja} = \delta_{ij}\,.
\end{equation}
Taking relations (\ref{FfGg}) into account, we see that (\ref{trFG}) 
holds providing that $H_1, H_2$ are unitary 
and condition (\ref{betatr}) is satisfied.
\end{proof} 

\begin{proof}[\bf Proof of Proposition~\ref{FGex1}]
$i)$ We have $H_1 = U$, $H_2 = U^*$.
The l.h.s. and the r.h.s. of (\ref{gamHH}) are equal,
respectively, to the $(12)$ and $(21)$ blocks of
$(U^* U)$ and hence they vanish identically. \\
$ii)$
It is straightforward to verify that $H_1, H_2$ are unitary and
 that both sides of  (\ref{gamHH}) vanish identically.
\end{proof} 

\begin{proof}[\bf Proof of Proposition~\ref{FGex2}]
It is straightforward to verify that $H_1, H_2$
are unitary providing that relations (\ref{Daa}) hold. Further, we have 
$G_{11} F_{12} + G_{12} F_{22} = 
 M Z_1^{\sigma_2 \circ \sigma_1} P_{\sigma_2} P_{\sigma_1}$ and
$G_{21} F_{12} + G_{22} F_{22} = -Z_2 \bar{M} P_{\sigma_2} P_{\sigma_1}$.
Therefore condition (\ref{ZZss}) implies equality~(\ref{gamHH}).
\end{proof} 

\begin{proof}[\bf Proof of Lemma~\ref{FixP}]   
For the sake of brevity, if $D$ is a diagonal matrix, we will write for 
its diagonal entries $D_i$ instead of $D_{ii}$. Recall that $D_1, D_2$ 
are non--singular.

If $\sigma_1, \sigma_2$ does not satisfy (\ref{ssss}), then 
$\sigma' \neq \sigma''$ and so there exist $i, j$ such that
$(P_{\sigma'})_{ij} =1$ but $(P_{\sigma''})_{ij} = 0$ and thus
the $(ij)$ matrix entry of the l.h.s. of (\ref{Deq3}) cannot vanish. 
 
Suppose that $\sigma_1, \sigma_2$ satisfy  (\ref{ssss}) but
$\sigma_1(i)=\sigma_2(i)=i$ for some~$i$.
Then we have $(P_{\sigma'})_{ii}= (P_{\sigma''})_{ii}=1$. 
Therefore the $(ii)$ matrix entry of the l.h.s. of (\ref{Deq3}) is 
$M_{ii} \equiv \bigl( |(D_1)_{i}|^2   +
 |(D_2)_{i}|^2 \bigr) (\bar{D}_1)_{i} (D_2)_{i} $ 
and so it cannot vanish.

If $\sigma_1$ and $\sigma_2$ are involutions
or they commute, then $\sigma_1, \sigma_2$ satisfy  (\ref{ssss}) 
and $\sigma' = \sigma'' =\sigma_2^{-1} \circ \sigma_1$.
Suppose that $(\sigma_2^{-1} \circ \sigma_1)(i)=i$ for some~$i$.
Then $(P_{\sigma'})_{ii}= (P_{\sigma''})_{ii}=1$.
Therefore the $(ii)$ matrix entry of the l.h.s. of (\ref{Deq3}) is 
$M_{ii} \equiv |(D_1)_{i}|^2 (\bar{D}_1)_{\sigma_1^{-1}(i)} 
 (D_2)_{\sigma_1^{-1}(i)}
 + |(D_2)_{i}|^2 (\bar{D}_1)_{\sigma_2^{-1}(i)} 
 (D_2)_{\sigma_2^{-1}(i)}$. 
Note that, for commuting $\sigma_1$ and $\sigma_2$,
equality $(\sigma_2^{-1} \circ \sigma_1) (i) =i$
implies that $\sigma_2^{-1} (i) = \sigma_1^{-1} (i)$.
The same is true if $\sigma_1$ and $\sigma_2$ are involutions.
Therefore,
$M_{ii} = \bigl(|(D_1)_{i}|^2 + |(D_2)_{i}|^2 \bigr) 
 (\bar{D}_1)_{\sigma_1^{-1}(i)} (D_2)_{\sigma_1^{-1}(i)}$
which cannot vanish.
\end{proof}  

\begin{proof}[\bf Proof of Proposition~\ref{SSn4}] 
$i)$ The group $\CS_4$ splits into five nonintersecting conjugacy
classes, $\CO_i$, $i=0,{\ldots},4$. 
For every two elements $\sigma_1, \sigma_2 \in \CO_i$, 
there exists $\tau \in \CS_4$ such that 
$\sigma_2 = \tau^{-1} \circ \sigma_1 \circ \tau$. 
$\CO_0$ contains only $\sigma=id$.
$\CO_1$ contains six involutions that have two fixed
points, e.g. $\sigma=(1)(23)(4)$.
$\CO_2$ contains eight elements that have one fixed
point and are of order three, e.g. $\sigma=(123)(4)$.
$\CO_3$ contains six elements that have no fixed
points and are of order four, e.g. $\sigma=(1234)$.
$\CO_4$ contains three involutions that have no fixed
points, e.g. $\sigma=(12)(34)$.
Without a loss of generality, we will take the mentioned 
above representatives of each conjugacy class $\CO_i$ as 
$\sigma_1$ and will search for all inequivalent 
admissible pairs $\sigma_1, \sigma_2$, where
$\sigma_2 \in \CO_j$, $j \geq i$.

For $\sigma_1=id$,  $\sigma_2$ must have no fixed points. 
We can take as $\sigma_2$ the mentioned above representatives 
of $\CO_3$ and~$\CO_4$.

For $\sigma_1=(1)(23)(4)$, the only suitable 
$\sigma_2$ from $\CO_1$ is $(14)(2)(3)$ because 
$\sigma_2^{-1} \circ \sigma_1$ must have no fixed points.
Further, note that $\sigma_2$ cannot be from $\CO_2$
because, in this case, equation (\ref{ssss})
would imply that $\sigma_1 \asymp \sigma_2^{-2}=\sigma_2$. 
However, the commutant of every $\sigma \in \CO_2$
consists only of $id,\sigma, \sigma^{-1}$.
The suitable elements from $\CO_3$ are 
$\sigma_2=(1342)$, $\sigma_2=(1243)$ and these from
$\CO_4$ are $\sigma_2=(12)(34)$, $\sigma_2=(13)(24)$.
In the each case, the corresponding admissible pairs are equivalent 
by the transformation (\ref{sstt}) with $\tau=\sigma_1$.
 
The consideration for $\sigma_1 \in \CO_2, \CO_3$
is similar and we omit its details.  
Finally, for $\sigma_1=(12)(34) $, $\sigma_2$ can be either of
the other two elements from $\CO_4$. The corresponding 
admissible pairs are equivalent by the transformation (\ref{sstt}) 
with $\tau=(1)(2)(34)$.

\vspace*{1mm}
$ii)$ We have $4 D_i \bar{D}_i = I_4$ and hence (\ref{trdd12a}) 
is satisfied and (\ref{Deq1})--(\ref{Deq2}) hold for $Q^2=8$. 
Note that $\sigma_2^{-1} \circ \sigma_1$ has no fixed points 
in all the cases except~h). So, (\ref{trdd12b}) is satisfied 
trivially except for the case h), where we have 
$\tr \bigl( D_1 \bar{D}_2 P_{\sigma_1} P_{\sigma_2}^t \bigr) =
e^{i\pi (u_2-v_2)}+e^{i\pi (u_4-v_4)}$.
Note also that $\sigma_1, \sigma_2$ satisfy the hypotheses of
Theorem~\ref{DDSS} in all the cases except h) and~f).
Therefore, in all these cases, it is sufficient to find 
$\vec{u}, \vec{v}$ that fulfil condition~(\ref{ABsub}). 
It is straightforward to check that suitable pairs of vectors can
be chosen as follows: $ {\vec u}=0$ for all the cases except c) and
\begin{eqnarray*}
& a), b): \ 4{\vec v}=(1,-1,1,-1);   \qquad
 c): \ 4{\vec u}=(0,1,-1,0),\ \ 4{\vec v}=(1,0,0,-1); & \\[0.5mm]
& d), e): \  2{\vec v}=(1,0,0,1);  \qquad
  f), g), h): \  {\vec v}= (1,1,0,0); \qquad
 i), j) : \   {\vec v}=(1,0,0,0).   &
\end{eqnarray*}
In the cases h) and f), one has to verify 
relation (\ref{Deq3}) by inspection.
\end{proof} 

\begin{proof}[\bf Proof of Proposition~\ref{sscomp}]
$\sigma_1, \sigma_2$ is an admissible pair because
$\sigma_1 \asymp \sigma_2$ and $ \sigma_2^{-1} \circ \sigma_1$
has no fixed points.  The latter property implies also that
(\ref{trdd12b}) is satisfied trivially. Set 
$D_1 = \frac{1}{\sqrt{n}} \mathrm{diag}
(e^{i\pi u_1},\ldots, e^{i\pi u_n})$ and
$D_2 = \frac{1}{\sqrt{n}} \mathrm{diag}
(e^{i\pi v_1},\ldots, e^{i\pi v_n})$, 
${\vec u}, {\vec v} \in {\mathbb R}^n$.
Then $n D_i \bar{D}_i = I_n$ and hence (\ref{trdd12a}) is 
satisfied and  (\ref{Deq1})--(\ref{Deq2}) hold for $Q^2=n^2/2$.  
Note that $\sigma_1, \sigma_2$ satisfy the hypotheses of
Theorem~\ref{DDSS}, hence it is sufficient to find 
$\vec{u}, \vec{v}$ that fulfil condition~(\ref{ABsub}).
Consider vectors ${\vec y}_i \in \mathrm{Ker} (I_n+P_{\sigma_i})$, $i=1,2$
such that  $({\vec y}_i)_k=0$ if $k$ is a fixed point of $\sigma_i$ and
$({\vec y}_i)_k= \pm 1 $ otherwise.
Clearly, there is an equal amount of $+1$ and $-1$
among the components of ${\vec y}_i$  corresponding to each
cycle in~$\sigma_i$. Since $\sigma_1$ and $\sigma_2$ have 
complementary sets of fixed points, we have 
a) $({\vec y}_1 + {\vec y}_2)_k= \pm 1$ for all~$k$;\ 
b) $P_{\sigma_2} {\vec y}_1 = {\vec y}_1$,  
$P_{\sigma_1} {\vec y}_2 = {\vec y}_2$, so that 
$A {\vec y}_1 = 4{\vec y}_1$,  $B {\vec y}_2 =4 {\vec y}_2$. 
Therefore ${\vec u} = \frac{1}{4} {\vec y}_1$
and ${\vec v} = \frac{1}{4} {\vec y}_2$  fulfil condition~(\ref{ABsub}).
\end{proof} 

\begin{proof}[\bf Proof of Theorem~\ref{DDSS}]
For the brevity of notations, let $e^{\vec{y}}$, 
where $\vec{y}\in {\mathbb C}^n$, stand for the
diagonal matrix $\mathrm{diag} (e^{y_1},\ldots,e^{y_n})$.
Then, for $D_1, D_2$ given by (\ref{d12unitr}),
we have $D_i \bar{D_i} = \mu^{-2} e^{2 \vec{x}}$
and therefore both equations (\ref{Deq1})--(\ref{Deq2})
are equivalent to the following one:
\begin{equation}\label{DPxQ} 
   \mu^{-2} e^{2 (I_n+P_{\sigma_1})\vec{x}} +
   \mu^{-2} e^{2 (I_n+P_{\sigma_2})\vec{x}} = Q^{-2} I_n .
\end{equation}
If (\ref{ppx}) is satisfied, then (\ref{DPxQ}) holds
and we have $Q^2= \mu^2/2$.

Since $\sigma_1, \sigma_2$ satisfy (\ref{zzeq0}), we have to
to verify that (\ref{dddd}) holds.
Substituting (\ref{d12unitr}) into  (\ref{dddd}), we obtain
the following equation:
\begin{equation}\label{DPxuv} 
   e^{ (P_{\sigma_1}+P_{\sigma_2})(I_n+P_{\sigma_1})\vec{x}
   + i\pi(A-P_{\sigma_2}^2) \vec{u} + 
   i\pi P_{\sigma_1}^2 \vec{v}}  
  =-  
  e^{ (P_{\sigma_1}+P_{\sigma_2})(I_n+P_{\sigma_2})\vec{x}
   - i\pi (B-P_{\sigma_1}^2) \vec{v} - 
   i\pi P_{\sigma_2}^2 \vec{u}} ,
\end{equation}
where $A, B$ are given by~(\ref{AB1}). 
If (\ref{ppx}) is satisfied, then (\ref{DPxuv})
is equivalent to equation $e^{i\pi\vec{w}} = -I_n$
which implies that all the components of $\vec{w}$
must be odd integers.

It remains to note that relations (\ref{trdd12a})--(\ref{trdd12b}) 
are satisfied thanks to the choice of $\mu$ in (\ref{d12unitr})
and the condition that $\sigma_2^{-1} \circ \sigma_1$
has no fixed points.
\end{proof}  

\begin{proof}[\bf Proof of Proposition~\ref{VVn4k}] 

$\zeta, z_1, z_2 \in {\mathbb C}$ satisfying (\ref{zzet1})
can be parametrized as follows: $\zeta=e^{i \pi \phi}$, 
$z_1 = \mu^{-1} e^{x+i \pi \alpha}$, 
$z_2 = \mu^{-1} e^{-x+i \pi \beta}$,
where $x,\mu,\alpha,\beta \in {\mathbb R}$ and
$\mu^2=n \cosh(2x)$. Therefore, $D_1, D_2$ are
given by (\ref{d12unitr}), where $\vec{x}=(x,-x,x,-x,\ldots)$,
$\vec{u}=(\alpha,\beta,\alpha,\beta,\ldots)$, and
$\vec{v}=\vec{u}+\vec{\rho}$, 
$\vec{\rho}=(0,\phi+1,0,\phi,\ldots)$. 
For $n=4l$ and $\sigma_1, \sigma_2$ given by 
(\ref{ssigma1}) or (\ref{ssigma2}), such
$\vec{x}$ satisfies~(\ref{ppx}) and, furthermore, we have
$(P_{\sigma_1} - P_{\sigma_2}) \vec{u}= 0$ and
$\vec{\rho}{\phantom{\,}}' \equiv 
(P_{\sigma_1} - P_{\sigma_2}) \vec{\rho}= (-1,0,1,0,\ldots)$.
Thus, for $A, B$ given by (\ref{AB1}), we have
$\vec{w} = A \vec{u} + B \vec{v} = B \vec{\rho}
= (I_n + P_{\sigma_1}) \vec{\rho}{\phantom{\,}}'$.
For the either choice of $\sigma_1$, all the components
of $\vec{w}$ are odd integers and so the hypotheses
of Theorem~\ref{DDSS} are satisfied.
\end{proof}

\begin{proof}[\bf Proof of Proposition~\ref{VVn4}]
For $D_1$, $D_2$ given by (\ref{D4ns}), we have 
$D_1 \bar{D}_1 D_1^{\sigma_1} \bar{D}_1^{\sigma_1}= D_0$,
$D_2 \bar{D}_2 D_1^{\sigma_2} \bar{D}_1^{\sigma_2}= 
  P_{\sigma_2} D_0$,
$D_1 \bar{D}_1 D_2^{\sigma_1} \bar{D}_2^{\sigma_1}= 
 P_{\sigma_1} D_0$, and
$D_2 \bar{D}_2 D_2^{\sigma_2} \bar{D}_2^{\sigma_2}=  
 P_{\sigma_1^{-1}} D_0$, where
$D_0 \equiv |z_1|^2 \,\mathrm{diag}_4
(|z_3|^2,|z_2|^2,|z_2|^2,|z_3|^2)$.
Since $(I_n+P_{\sigma_2}) D_0 =
(P_{\sigma_1}+P_{\sigma_1^{-1}}) D_0 = Q^{-2} I_n$,
where $Q$ is given by~(\ref{Q4ns}), we conclude that
equations (\ref{Deq1})--(\ref{Deq2}) hold.
Since $\sigma_1, \sigma_2$  satisfy (\ref{zzeq0}), it suffices 
to verify~(\ref{dddd}). A direct computation yields
$D_1^{\sigma_2} \bar{D}_1^{\sigma_2 \circ \sigma_1} 
  D_2^{\sigma_1 \circ \sigma_1} \bar{D}_1^{\sigma_1} =
|z_1|^2 {\rm diag}_4\,(z_2 \bar{z}_3, \zeta \bar{z}_2 \bar{z}_3,
  \bar{z}_2 z_3,-\zeta \bar{z}_2 \bar{z}_3)=
 - D_2^{\sigma_2} \bar{D}_1^{\sigma_2 \circ \sigma_2} 
  D_2^{\sigma_1 \circ \sigma_2} \bar{D}_2^{\sigma_1}$,
so that (\ref{dddd}) holds.

It remains to note that (\ref{trdd12a}) holds thanks
to the condition (\ref{zzet2}) whereas (\ref{trdd12b})
is equivalent to the condition 
$\sum_{k=1}^{2l} (D_1 \bar{D}_2)_{2k,2k} = 0$ which
also holds as seen from~(\ref{D4ns}).
\end{proof} 

\vspace*{1mm}
\small{
{\bf Acknowledgements.} 
The author is grateful to P. Kulish for useful discussions.
This work was supported in part by the grant MODFLAT of the 
European Research Council (ERC) and by the NCCR SwissMAP 
of the Swiss National Science Foundation,
 and in part by the Russian Fund for Basic Research 
grants  14-01-00341 and 13-01-12405-ofi-m. 
}

\end{document}